\documentclass[aps,prd,reprint,nofootinbib,superscriptaddress]{revtex4-2}

\usepackage[utf8]{inputenc}
\usepackage[T2A,T1]{fontenc}
\usepackage[main=english]{babel}
\AtBeginDocument{\DeclareFontFamilySubstitution{T2A}{\rmdefault}{Tempora-TLF}}
\def\Bogolyubov{\textbf{\foreignlanguage{russian}{Б}}}

\usepackage[a4paper,margin=2cm]{geometry}

\usepackage{amsthm}
\newtheorem*{proposition}{Proposition}
\newtheorem*{theorem}{Theorem}

\usepackage{anyfontsize}
\usepackage[varvw]{newtx}
\usepackage{microtype}

\usepackage[dvipsnames]{xcolor}
\usepackage[hidelinks,colorlinks=true,allcolors=Blue]{hyperref}
\usepackage{orcidlink}
\usepackage{aas_macros}

\usepackage{mathtools}
\usepackage{braket}
\newcommand{\ketbra}[2]{\ket{#1}\!\!\bra{#2}}
\usepackage{tensor}
\let\t\tensor

\def\ee{\mathrm e}	% Euler’s constant
\def\ii{\mathrm i}	% Imaginary unit
\def\dd{\mathrm d}	% Leibniz’s differential
\def\half{\tfrac{1}{2}}

\let\bs\boldsymbol
\def\ha{\hat a}

\def\LieD{\mathfrak L}	% Lie derivatie
\def\Metric{g}
\def\metric{l}
\def\lapse{\zeta}

\def\shift{\xi}

\def\Killing{K}

\def\MMM{\textsc{mmm}}
\def\AB{\textsc{a{\tiny /}b}}
\def\alice{\textsc{a}}
\def\bob{\textsc{b}}
\def\charlie{\textsc{c}}

\def\sys{\text{sys}}
\def\aux{\text{aux}}
\def\corr{\text{cor}}

\def\X{\hat{\mathbb X}}
\def\Y{\hat{\mathbb X}}

\def\identity{\mathrm I}

\def\transpose{\textsc t}

\def\channel{\mathcal N}

\DeclareMathOperator{\fidelity}{\mathcal F}
\DeclareMathOperator{\trace}{tr}

\DeclareMathOperator{\linearHull}{span}

\usepackage[capitalize]{cleveref}

\hypersetup{
    pdftitle={Gravitational redshift as a quantum channel: modeling the effects of gravitational redshift in quantum optics},
    pdfauthor={Thomas Mieling, Andreas Wolfgang Schell, David Edward Bruschi},
}

\begin{document}

% \date{\today}
\title{Gravitational redshift as a quantum channel: \\
modeling the effects of gravitational redshift in quantum optics}

\author{Thomas Mieling \orcidlink{0000-0002-6905-0183}}
\email{t.mieling@fz-juelich.de}
\affiliation{Institute for Quantum Computing Analytics (PGI-12), Forschungszentrum Jülich, 52425 Jülich, Germany}
\author{Andreas Wolfgang Schell \orcidlink{0000-0003-0849-9558}}
\affiliation{Division of Light-Matter-Interaction, Johannes Kepler Universität Linz, 4040 Linz, Austria}
\author{David Edward Bruschi \orcidlink{0000-0002-3816-5439}}
\email{david.edward.bruschi@posteo.net}
\email{d.e.bruschi@fz-juelich.de}
\affiliation{Institute for Quantum Computing Analytics (PGI-12), Forschungszentrum Jülich, 52425 Jülich, Germany}
\affiliation{Theoretical Physics, Universität des Saarlandes, 66123 Saarbrücken, Germany}

\begin{abstract}
    The gravitational frequency shift of light is well understood in the theory of classical electromagnetism. Nevertheless, its description in quantum theory is not yet fully developed.
    Recent work pointed out inconsistencies in previously developed models aimed at describing the gravitational redshift as an effective multi-mode mixer (\MMM) acting on modes of light, but so far a complete solution of these issues was not obtained. Here, we identify the root cause of the \MMM\ model’s inconsistency and provide two complementary approaches to correct it: a  “natural” one from a field-theoretic perspective, and another adapted to the language of quantum mechanics of finite-dimensional systems.
    We show that the second approach allows for modeling of the redshift in a multi-mode transmission setup as a quantum channel that can be characterized using standard quantum information-theoretic techniques when restricting the input states to Gaussian states of light.
\end{abstract}

\maketitle

The gravitational redshift is a well-known prediction of Einstein’s theory of gravity \cite{1908JRE.....4..411E}, which has been experimentally verified in numerous setups \cite{1960PhRvL...4..337P,2020A&A...643A.146G,2023NatCo..14.4886Z}.
Within the framework of general relativity and classical field theory in curved space-times, explicit predictions for frequency shifts are typically derived using the laws of geometrical optics \cite{Wald:1984,2006AmJPh..74.1017H} or classical field equations such as the (massless) Klein--Gordon equation \cite{1998CQGra..15.1653M} as well as Maxwell's equations \cite{2002CQGra..19..571M,2003CQGra..20...85K}. The theoretical description of the gravitational redshift in the quantum realm, on the other hand, is the subject of ongoing research. Specifically, recent work modeled gravitational frequency shifts of single photons as effective linear multi-mode mixers (\MMM) acting on a set of ladder operators associated to various frequency spectra \cite{2023AnP...53500468B,2023JPhCS2531a2016A}. Such \MMM\ models were recently challenged in a study that pointed out the fact that the resulting transformation matrices are---contrary to what was expected in the original proposal---not unitary \cite{Leber_2026}.
While the central predictions derived from the \MMM\ model are indeed plausible, particularly in the regime of small redshifts, this criticism left the model without a firm mathematical basis.
Here, we demonstrate that the reason for this inconsistency lies in the fact that the frequency shift operator has no regular invariant subspaces of finite dimension, which provides an obstruction for constructing matrix representations.

We then present two alternatives to the current form of the \MMM\ model: (i) one approach working without matrix representations entirely, and (ii) one approach based on a unitary dilation of the non-unitary operator obtained by compressing the redshift operator to a finite subspace of mode functions.
The latter approach naturally allows for the application of standard quantum-optical and quantum-information-theoretic techniques within the covariance matrix formalism when restricting the input states to Gaussian states. In particular, considering multi-mode transmission protocols as done in the literature~\cite{2023AnP...53500468B,2023JPhCS2531a2016A}, this technique allows modeling the redshift as a quantum channel.
In turn, this naturally provides an efficient way to characterize the process from a theoretical perspective, and also suggests experimental setups for measuring its properties.

This paper is structured as follows:
\cref{s:wave packet transformations} describes the action of gravitational frequency shifts of wave packets with a finite spectral bandwidth and in \cref{s:MMM presentation and limits} we review the central idea of the \MMM\ model and prove a no-go theorem providing an obstruction for carrying out constructions as proposed in the original \MMM\ papers \cite{2023AnP...53500468B,2023JPhCS2531a2016A}. 
\Cref{s:compressions and dilations} then presents two alternative methods for correcting the \MMM\ model – one of which turns out to be a minor modification in the spirit of the recent article \cite{Leber_2026}.
The utility of this model is demonstrated in \cref{s:quantum channel}, where we explicitly construct the quantum channel associated to the transmission of Gaussian quantum states in a multi-mode setting. Finally, \cref{s:MMM review of results} assesses the \MMM\ model’s predictions on the basis of our refined model, and \cref{s:conclusions} closes with a summary and outlook of future applications of our findings.

\section{Wave-packet transformations}
\label{s:wave packet transformations}

In this work we focus on particles propagating within stationary space-times. 
A stationary gravitational field is described by a time-oriented Lorentzian four-manifold $(M, \t\Metric{_\mu_\nu})$ with a timelike Killing vector field $\t\Killing{^\mu}$. Locally, such space-times admit coordinates $(\t x{^\mu}) = (t, \t x{^i})$ relative to which $\t\Killing{^\mu} = (1, 0, 0, 0)$ and the metric takes the form
\begin{align}
	\Metric
		= - \lapse^2 (c \dd t - \t\shift{_i} \t{\dd x}{^i})^2
		  + \t\metric{_i_j} \t{\dd x}{^i} \t{\dd x}{^j},
\end{align}
where $c$ is the speed of light in vacuum, and the lapse $\lapse$, shift $\t\shift{_i}$, as well as the spatial metric $\t\metric{_i_j}$ are independent of the temporal coordinate $t$. For a detailed discussion of the geometric meaning of these quantities, we refer the reader to appropriate literature~\cite{2025PhRvR...7a3162M,2025PhRvA.112f2211M}.

Particles and photons are modeled in this work as localized excitations of (quantum) fields. In particular, we will be interested in the effects of gravitational redshift on the quantum state of photons exchanged between two sites as illustrated in \cref{fig:schematic AB}. Therefore, characterizing relevant properties of the fields of interest will allow us to model the effects of interest.  

In this context, a tensor field $\phi$ is said to be monochromatic if it is of the form $\phi(t, x) = \varphi(x) \ee^{- \ii \t \omega{_*} t}$, or equivalently if $\LieD_\Killing \phi = - \ii \t \omega{_*} \phi$, where $\LieD$ denotes the Lie derivative.
The frequency $\t \omega{_*}$ will henceforward be referred to as the Killing frequency, which is related to the frequency $\t\omega{_\alice}$ relative to a stationary observer at position $\t p{_\!_\alice}\in M$ by
\begin{align}
	\t\omega{_*} = \t\omega{_\alice} \t\lapse{_\alice},
\end{align}
where $\t\lapse{_\alice} = \lapse(\t p{_\!_\alice})$ is the lapse at the location $\t p{_\!_\alice}$.
This can be seen, for example, using the notion of generalized Killing vector fields in the sense of Ref.~\cite{2008CQGra..25w5020H}.
The observer’s four-velocity $\t u{^\mu}$ defines a unique such field $\t*\Killing{^\prime^\mu}$ by the requirements that $\t*\Killing{^\prime^\mu} = \t u{^\mu}$ and ${\t\nabla{_\mu} \t*\Killing{^\prime_\nu} - \t\nabla{_\nu} \t*\Killing{^\prime_\mu} = 0}$ hold \emph{along} the observer’s world-line. In the concrete setup considered here, one has $\t\Killing{^\prime^\mu} = \t\Killing{^\mu} / \t\lapse{_\alice}$ and the Lie derivative of $\phi$ along this field takes the form $\t\LieD{_{\Killing'}} \phi = (1/\t\lapse{_\alice}) \t\LieD{_\Killing} \phi = - \ii (\t\omega{_*} / \t\lapse{_\alice}) \phi \equiv - \ii \t\omega{_\alice} \phi$, as claimed.
As a consequence, the frequencies for two observers at positions $\t p{_\!_\alice}$ and $\t p{_\!_\bob}$ are related by
\begin{align}
	\t\omega{_\alice} \t\lapse{_\alice} = \t\omega{_\bob} \t\lapse{_\bob},
\end{align}
which is the well-know relation between frequencies as measured by two different observers in a general stationary space-time \cite{Wald:1995}.

\begin{figure}
    \centering
    \includegraphics[width=\columnwidth]{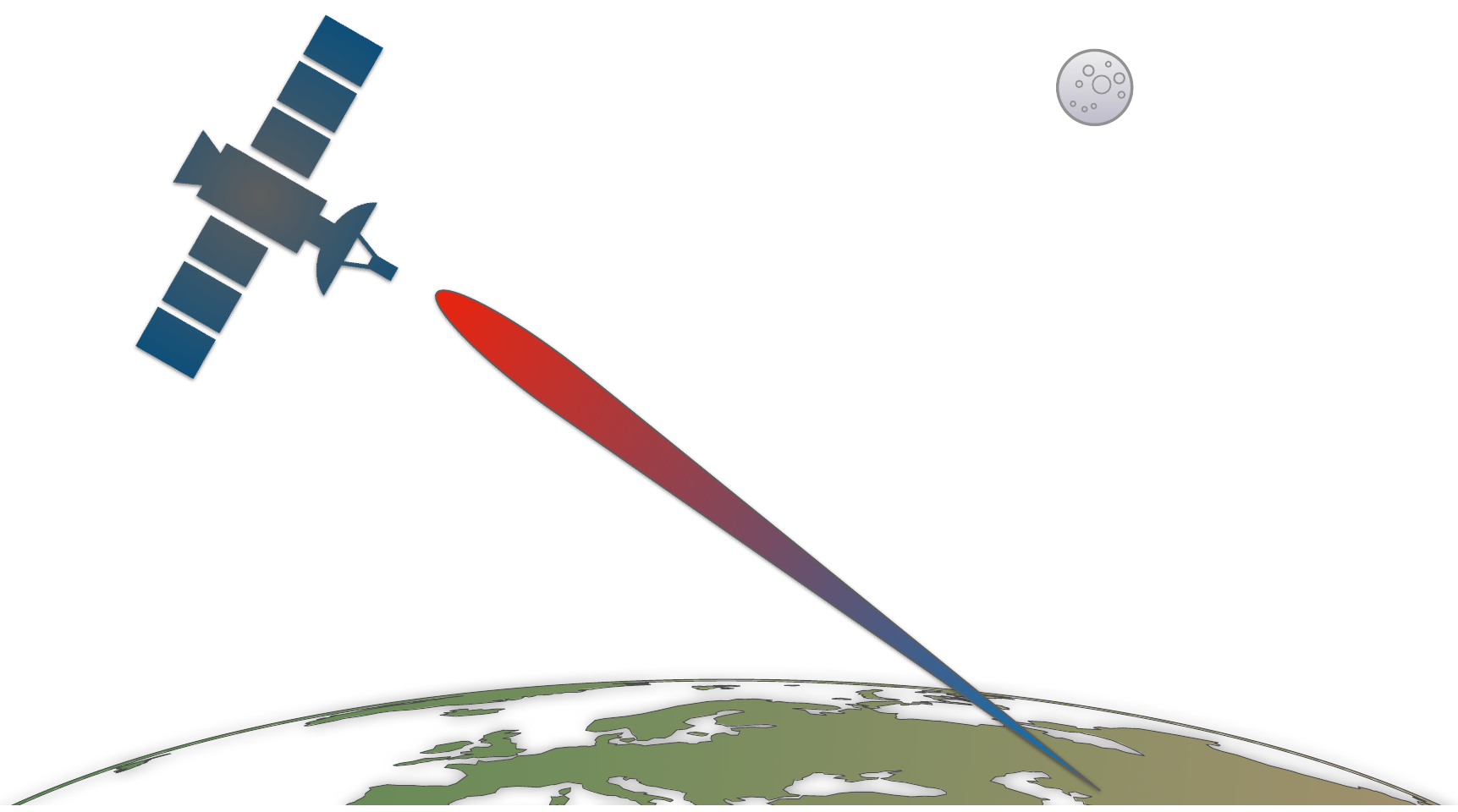}
    \caption{
        Schematic illustration of the generic setup considered in this work, where photons are exchanged between stationary observers in a time-independent gravitational field.
        Note that for satellite applications, the present analysis is restricted to geostationary orbits.
    }\label{fig:schematic AB}
\end{figure}

Such monochromatic fields typically do not have finite norms, as is standard for sharp-frequency excitations in field theory \cite{Srednicki:2007}, while wave-packets constructed by Fourier synthesis of such sharp modes do \cite{Itzykson:Zuber:1980}. In the following, we consider finite-norm wave packets of the form
\begin{align}
	\label{eq:wave packet from Killing distribution}
	\t\psi{_f}(t, x)
		= \int_0^\infty \dd \t \omega{_*}\, f(\t \omega{_*}) \phi_{\t \omega{_*}}(t, x),
\end{align}
where, for each $\t\omega{_*} > 0$, $\phi_{\t \omega{_*}}$ is a monochromatic solution with Killing-frequency $\t \omega{_*}$, and we assume that the Klein--Gordon product of two such wave packets takes the form
\begin{align}
    (\t\psi{_f}, \t\psi{_h})
	&= \int_0^\infty \dd \t \omega{_*} \, \overline{f(\t \omega{_*})} h(\t \omega{_*}).
\end{align}
The complex function $f$ in \cref{eq:wave packet from Killing distribution} encodes the contributions of the various Killing frequencies $\t\omega{_*}$ and can be translated to a distribution function of the frequencies $\t\omega{_\alice}$ relative to a stationary observer at position $\t p{_\!_\alice}$ by 
\begin{align}
    \psi_f(t, x)
        = \int_0^\infty \dd \t\omega{_\alice} \, \t\lapse{_\alice} f(\t\omega{_\alice} \t\lapse{_\alice}) \phi_{\t\omega{_\alice} \t\lapse{_\alice}}(t, x).
\end{align}
This equation suggests identifying $\t\lapse{_\alice} f(\t\omega{_\alice} \t\lapse{_\alice})$ with the frequency spectrum relative to the observer at $\t p{_\!_\alice}$. However, it is advantageous to use, instead, a rescaled version, namely
\begin{align}
	\label{eq:spectrum transformation}
	\t f{_\alice}(\t\omega{_\alice})
		= \sqrt{\t\lapse{_\alice}}\, f(\t\omega{_\alice} \t\lapse{_\alice}).
\end{align}
The reason for this particular choice is that inner products can then be written as
\begin{align*}
		\int_0^\infty \hspace{-0.75em} \dd \t\omega{_*} \, \overline{f(\t\omega{_*})} h(\t\omega{_*})
		&
		= \int_0^\infty \hspace{-0.75em} \dd \t\omega{_\alice} \, \t\lapse{_\alice} \overline{f(\t\omega{_\alice} \t\lapse{_\alice})} h(\t\omega{_\alice} \t\lapse{_\alice})
		\\
        &= \int_0^\infty \hspace{-0.75em} \dd \t\omega{_\alice} \, \overline{\sqrt{\t\lapse{_\alice}} f(\t\omega{_\alice} \t\lapse{_\alice})} \sqrt{\t\lapse{_\alice}} h(\t\omega{_\alice} \t\lapse{_\alice})
		\\
        &= \int_0^\infty \hspace{-0.75em} \dd \t\omega{_\alice} \, \overline{\t f{_\alice}(\t\omega{_\alice})} \t h{_\alice}(\t\omega{_\alice}),
\end{align*}
which implies that for any two stationary observers at positions $\t p{_\!_\alice}$ and $\t p{_\!_\bob}$ one has
\begin{align}
	\label{eq:form-invariance of inner products}
	\int_0^\infty \dd \omega \, \overline{\t f{_\alice}(\omega)} \t h{_\alice}(\omega)
	= \int_0^\infty \dd \omega \, \overline{\t f{_\bob}(\omega)} \t h{_\bob}(\omega).
\end{align}
In turn, as a consequence of \cref{eq:spectrum transformation}, one concludes that the two distribution functions $\t f{_\alice}$ and $\t f{_\bob}$ are related by
\begin{align}
	\t f{_\alice}(\t\omega{_\alice}) / \sqrt{\t\lapse{_\alice}}
		= f(\t \omega{_*})
		= \t f{_\bob}(\t\omega{_\bob}) / \sqrt{\t\lapse{_\bob}}.
\end{align}
Therefore, for all $\omega > 0$ one has
\begin{align}
	\label{eq:spectrum bob from alice}
	\t f{_\bob}(\omega) = \sqrt{\t\lapse{_\bob}/\t\lapse{_\alice}}\, \t f{_\alice}(\omega \t\lapse{_\bob} / \t\lapse{_\alice}).
\end{align}
Defining the operator $\t U{_\chi}$ as
\begin{align}
	\label{eq:U operator definition}
	(\t U{_\chi} f)(\omega)
		:=
		\chi^{-1} f(\chi^{-2} \omega) ,
\end{align}
one thus has $\t f{_\bob} = \t U{_\chi} \t f{_\alice}$ with $\chi: = \sqrt{\t\lapse{_\alice} / \t\lapse{_\bob}}$.\footnote{The notation $\chi$ for the redshift parameter is not standard, but is used here to keep consistency with the works in the related literature \cite{2014PhRvD..90d5041B,Leber_2026}.}
This result was originally formulated in the first work on this topic \cite{2014PhRvD..90d5041B}, albeit using a less general approach, and has also been re-derived in the recent article \cite{Leber_2026}.
The transformation $\t U{_\chi}$ is at the center of the analysis carried out in this work.

\Cref{eq:form-invariance of inner products} can be rephrased as saying that $\t U{_\chi}$ is a unitary transformation on $L^2(\mathbf R_+)$, since $(f, h) = (U f, U h)$ is the defining property of a unitary operator $U$.
Clearly, the operator $\t U{_\chi}$ satisfies the group composition relation
\begin{align}
    \label{eq:U composition}
	\t U{_{\chi}} \t U{_{\chi'}} = \t U{_{\chi \chi'}},
\end{align}
which is consistent with the expectation that a composition of redshift processes from $\t p{_\!_\alice}$ to $\t p{_\!_\bob}$ and from $\t p{_\!_\bob}$ to $\t p{_\!_\charlie}$ is equivalent to a single redshift process from $\t p{_\!_\alice}$ to $\t p{_\!_\charlie}$.
Note that $\t U{_{\chi=1}} = \identity$ (the identity operator), as well as
\begin{equation}
    \t U{_{\chi}^{-1}} = \t U{_{1/\chi}}.
\end{equation}
So far, the factor $\chi$ was introduced by considering the redshift of monochromatic fields, and the construction of $\t U{_\chi}$ that was consequently used to describe the redshift of wave packets could appear to be based primarily on formal rather than physical arguments \cite{2023AnP...53500468B}. The result of this construction, however, can be made physically tangible by studying the action of $\t U{_\chi}$ on the moments of a wave packet’s frequency distribution.
Indeed, by defining the $k$-th moment $\t\mu{_k}$ of a complex spectral function $f$ as
\begin{align}
	\label{eq:k-th moment def}
	\t\mu{_k}[f] &= \frac{
		\int_0^\infty \! \dd \omega\, \omega^k\, |f(\omega)|^2
	}{
		\int_0^\infty \! \dd \omega\, |f(\omega)|^2
	},
\end{align}
it is immediate to verify that
\begin{align}
	\label{eq:moment transformation}
	\t\mu{_k}[\t U{_\chi} f]
		= \chi^{2k} \t\mu{_k}[f].
\end{align}
Hence, the gravitational frequency shift transforms all moments and cumulants in a coherent manner.
The transformation for $k = 1$ was first highlighted in the recent article \cite{Leber_2026}.
More generally, \cref{eq:moment transformation} implies that both the moment-generating function
$F_{\textsc m}(\lambda) = \sum_{k \in \mathbf N} \t\mu{_k} \lambda^k / k!$
and the cumulant-generating function
$F_{\textsc c}(\lambda) = \ln F_{\textsc m}(\lambda) = \sum_{k \in \mathbf N} \t c{_k} \lambda^k / k!$
transform as
\begin{equation}
    F[\t U{_\chi} f](\lambda) = F[f](\chi^2 \lambda).
\end{equation}

A quantity of significant interest is the quantity $\Delta$, defined as the inner product a spectrum $f$ with its transformed version,  which was introduced in the literature as a way to quantify the overlap of single-photon states \cite{2021PhRvD.104h5015B} as
\begin{equation}\label{delta:definition}
    \Delta_f(\chi) := (f, \t U{_\chi} f).
\end{equation}
This quantity is central to quantum communication protocols involving entanglement swapping \cite{2014PhRvD..90d5041B} as well as other interferometry-based tasks \cite{2023AnP...53500468B}. More recently, this quantity was studied in the perturbative regime of small redshifts, where a deeper understanding of the role of its real and imaginary parts was achieved \cite{Leber_2026}.
We define ${\rho_f(\chi) := |\Delta_f(\chi)|}$, and note that it is immediate to see that ${\rho_f(\chi) = \rho_f(\chi^{-1})}$. If $\rho_f$ admits a Laurent expansion
\begin{equation}
    \rho_f(\chi)
        = \sum_{n=-\infty}^{+\infty} \frac{\Delta_n}{|n|!}\chi^n,
\end{equation}
the quantities $\Delta_n$ are uniquely determined by $f$. Using the fact that $\rho_f(\chi=1)=1$, it is immediate to see that $\rho_f(\chi)$ can be re-written as
\begin{equation}
    \rho_f(\chi)
        = 1
          + \sum_{n=1}^{+\infty} \frac{\Delta_n}{n!}\frac{(\chi^n-1)^2}{\chi^{n}}.
\end{equation}
This expression can also be cast in the form $\rho_f(\chi) = 1 + 4\sum_{n=1} \frac{\Delta_n}{n!}\sinh^2(n r)$ by introducing $r = \half \ln \chi$.

\section{The \texorpdfstring{\MMM}{MMM} model and its limits}
\label{s:MMM presentation and limits}

The key idea of the \MMM\ model, in its original formulation~\cite{2023JPhCS2531a2016A,2023AnP...53500468B}, is that the action of $\t U{_\chi}$ on a \emph{finite} set of spectral distribution functions $\{\t f{_i}\}_{i = 1}^n$ should be expressible in terms of a \emph{finite-dimensional} matrix $\t u{_\chi}$ with components $\t u{_I_J}(\chi) = (\t f{_I}, \t U{_\chi} \t f{_J})$, with $I, J \in \{1, \ldots, n, \perp\}$, where $\t f{_\perp}$ is a suitably chosen auxiliary “perpendicular mode.”

Recently, it has been pointed out that the \MMM\ model cannot be universally correct~\cite{Leber_2026}. This conclusion was reached by determining limiting behaviors of the matrix elements $\t u{_I_J}(\chi)$ for $\chi \to 0$ and $\chi \to \infty$, which turn out to be incompatible with the matrix $\t u{_\chi}$ being unitary.
It is the aim of this section to explain the cause for this phenomenon.

The fundamental source of inconsistency of the \MMM\ model is that the unitary operator $\t U{_\chi}$ has no regular finite-dimensional invariant subspace whenever $\chi \neq 1$. That is, there is no finite set of well-behaved functions $\t f{_I}$ and coefficients $\t u{_I_J}$ such that $\t U{_\chi} \t f{_I} = \t\sum{_J} \t f{_J} \t u{_J_I}$.
Since the ladder operators $\ha^\dagger(f)$ on Fock space depend linearly on the mode functions \cite{Derezinski_Gerard}, a similar statement applies to them.
This issue precludes the applicability of the \MMM\ model, as it requires the existence of such invariant subspaces to enable the redshift to transform \emph{finitely many} ladder operators $\hat{a}_j$, collected in the vector
$\X = (\t{\ha}{_1}, \ldots, \t{\ha}{_n}, \t{\ha}{_\perp})^\transpose$, into linear combinations of the \emph{same} operators them via the map $\X \mapsto \t*u{^\dagger} \X $, as  previously assumed, see Eq.~(8) of Ref.~\cite{2023AnP...53500468B} and Eq.~(39) of Ref.~\cite{2023JPhCS2531a2016A}.
We will give a rigorous demonstration of this claim below.

Before presenting our theorem and proof, we wish to point out the precise point at which the prior literature \cite{2023AnP...53500468B,2023JPhCS2531a2016A,Leber_2026} implicitly made the assumption of the existence of invariant subspaces.\footnote{The aforementioned papers worked with smeared ladder operators, but since their commutation relations are determined by the Klein–Gordon products of their mode functions, we here chose to work with normalized spectral functions $f$ for notational simplicity.}
Starting with a single distribution function $f$, a natural question is: how similar is the shifted profile $\t U{_\chi} f$ to the original profile $f$, where similarity is measured using the Klein–Gordon product, and how does the deviation look like?
The answer can be directly obtained by decomposing $\t U{_\chi} f$ as
\begin{align}
	\label{eq:U f decomposition}
	\t U{_\chi} f
		= \underbrace{(f, \t U{_\chi} f) f}_\text{proportional to $f$}
		+ \underbrace{[\t U{_\chi} f - (f, \t U{_\chi} f) f]}_\text{orthogonal to $f$}.
\end{align}
Generically, the second term is non-zero.
Denoting by $N_\perp := \| \t U{_\chi} f - (f, \t U{_\chi} f) f \| = \sqrt{1 - \rho_f(\chi)^2}$ its norm and setting
\begin{align}
    \label{eq:h definition}
    h =  \frac{1}{N_\perp} \, [\t U{_\chi} f - (f, \t U{_\chi} f) f],
\end{align}
one thus has
\begin{align}
	\label{eq:redshift f h decomposition}
	\t U{_\chi} f
		= (f, \t U{_\chi} f) f + (h, \t U{_\chi} f) h.
\end{align}
So far, this line of reasoning is compatible with that of the original work on this topic~\cite{2023AnP...53500468B,2023JPhCS2531a2016A,Leber_2026}.
However, contrary to what was concluded there, \cref{eq:redshift f h decomposition} does \emph{not} imply that there is a $2\times2$-matrix $\t u{_\chi}$ such that
\begin{align}
    \label{eq:closure 2x2}
	\begin{pmatrix}
		\t U{_\chi} f \\
		\t U{_\chi} h
	\end{pmatrix}
	=
    \t*u{_\chi^\transpose}
	\begin{pmatrix}
		f\\
		h
	\end{pmatrix}
    .
\end{align}
To give a counterexample, we proceed as follows.
If \cref{eq:closure 2x2} were to hold, then not only $\t U{_\chi} f$ would lie in the span of $f$ and $h$, but so would $\t U{_\chi} h$, and hence also $\t*U{_\chi^{-1}} f$. Since $f$ and $h$ are orthonormal and $\t U{_\chi}$ is unitary, the following two parameters would be  equal $1$ for all $\chi$:
\begin{subequations}
\begin{align}
    \label{eq:MMM parameter 1}
    \gamma_1 &:= |(f, \t U{_\chi} h)|^2 + |(h, \t U{_\chi} h)|^2,
    \\
    \label{eq:MMM parameter 2}
    \gamma_2 &:= |(f, \t*U{_\chi^{-1}} f)|^2 + |(h, \t*U{_\chi^{-1}} f)|^2.
\end{align}
\end{subequations}
However, as shown in \cref{fig:MMM_accuracy}, the constraints $\t\gamma{_1} = 1$ and $\t\gamma{_2} = 1$ are violated for Gaussian spectra whenever $\chi \neq 1$.
This finding confirms the result of the literature~\cite{Leber_2026}, that is, the violation of the unitarity relation \eqref{eq:closure 2x2} occurs at order $O((\chi-1)^2)$.
\begin{figure}
    \centering
    \includegraphics{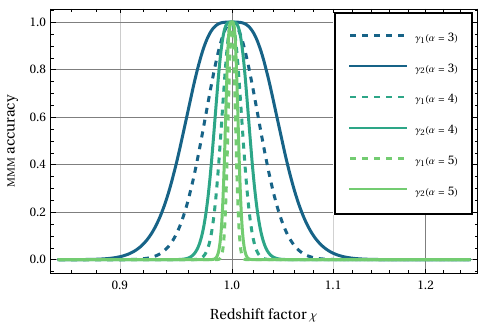}
    \caption{
        Plots of the parameters $\t\gamma{_1}$ and $\t\gamma{_2}$, defined in \cref{eq:MMM parameter 1,eq:MMM parameter 2}, that quantify the violation of the \MMM\ model.
        Here, $\alpha = \ln(\omega_0 / \sigma)$
        parametrizes the relative width of the Gaussian frequency spectra,
        where $\omega_0$ is the central frequency and $\sigma$ is the spectral width.
        Clearly, the \MMM\ prediction $\t\gamma{_1} = \t\gamma{_2} = 1$ is violated for all $\chi \neq 1$, with deviations from these values increasing with $\chi - 1$ and the “sharpness” of the Gaussian wave packets.
    }
    \label{fig:MMM_accuracy}
\end{figure}
More generally, it is clear that an equation of the form
\begin{align}  
	\label{eq:finite-dimensional representation}
	\begin{pmatrix}
		\t U{_\chi} \t f{_1}\\
		\vdots\\
		\t U{_\chi} \t f{_m}
	\end{pmatrix}
	=
    \t*u{_\chi^\transpose}
	\begin{pmatrix}
		\t f{_1}\\
		\vdots\\
		\t f{_m}
	\end{pmatrix}
\end{align}
implies that the spectra $\t f{_I}$ (which may comprise distributions $\t f{_i}$ of actual interest as well as any finite collection of auxiliary “perpendicular modes” $f_{\perp_k}$) span a subspace that is invariant under $\t U{_\chi}$ as it requires the action of $\t U{_\chi}$ on any linear combination of these $\t f{_I}$ to produce another linear combination of the same modes. Conversely, if $\t U{_\chi}$ has an invariant subspace of finite dimension, one can choose a finite orthonormal basis $\t f{_I}$ of this space and determine the matrix $\t u{_\chi}$ by computing the inner products $(\t f{_I}, \t U{_\chi} \t f{_J})$.

This argument shows that a formula of the kind \eqref{eq:finite-dimensional representation} is applicable if and only if $\t U{_\chi}$ has a finite-dimensional invariant subspace.
However, the following result shows that such spaces---if they exist---are necessarily
“irregular” in the following sense:
A space $W$ of spectra is said to be \emph{regular} if it has a basis $\{\t f{_i}\}$ such that for some $k > 0$, all $\t f{_i}$ have finite moments of order $k$, i.e., $\t\mu{_k}(\t f{_i}) < \infty$.
Conversely, $W$ will be said to be \emph{irregular} if there exits a basis $\{\t f{_i}\}$ such that for every $k > 0$ there is a basis element $\t f{_i}$ with $\t\mu{_k}(\t f{_i}) = \infty$.
With this definition at hand, we have the following result.

\begin{proposition}
	If, for any $\chi \neq 1$, $\t U{_\chi}$ leaves a finite-dimensional subspace $W \subset L^2(\mathbf R_+)$ invariant,
	then $W$ is irregular.
\end{proposition}
\begin{proof}
	If $W$ is a finite-dimensional subspace that is invariant under $\t U{_\chi}$,
	the finite-dimensional spectral theorem \cite[Sec. 8.5, Theorem 22]{1971_Hoffmann_Kunze} implies that $W$ admits an orthonormal basis of eigenvectors: $\t U{_\chi} \t f{_i} = \t \lambda{_i} \t f{_i}$
    (in fact, all $\t\lambda{_i}$ are unimodular, but this property is not of further use here).
	Now, by scaling-invariance of the moment functions $\t\mu{_k}$ and \cref{eq:moment transformation}, one has
	\begin{align}
		\label{eq:moment transformation eigenvectors}
		\t\mu{_k}[\t f{_i}]
			= \t\mu{_k}[\t\lambda{_i} \t f{_i}]
			= \t\mu{_k}[\t U{_\chi} \t f{_i}]
			= \chi^{2k} \t\mu{_k}[ \t f{_i}].
	\end{align}
	Since $\chi \neq 1$ by assumption, all moments $\t\mu{_k}[\t f{_i}]$  with $k > 0$ are infinite.
    Hence, $W$ is irregular.
\end{proof}

This result implies that finite collections of “well-behaved” mode spectra cannot span invariant subspaces of $\t U{_\chi}$ unless $\chi = 1$, which is in general physically uninteresting. This observation is made precise by the following result.

\begin{theorem}
	If $\mathcal{B}=\{\t f{_i}\}$ is any finite set of frequency spectra with finite moments of some order $k > 0$,
	then their span is invariant under $\t U{_\chi}$ if and only if $\chi = 1$.
\end{theorem}
\begin{proof}
	Since $\chi = 1$ implies $\t U{_\chi} = \identity$, the only non-trivial step is to prove that the invariance of $\linearHull(\mathcal{B})$ implies $\chi = 1$.
	As reasoned before, if $\linearHull(\mathcal{B})$ is invariant one can construct an orthonormal basis $\mathcal{B}':=\{f{^\prime_i}\}$ of mode functions $f{^\prime_i}$ that satisfy the eigenvalue equation $\t U{_\chi} \t*f{^\prime_i} = \t \lambda{_i} \t f{^\prime_i}$.
	Now, the Cauchy–Schwarz inequality implies
	\begin{align}
		\t\mu{_k}(\t*f{^\prime_i})
			\leq \left(
				\sum_j \t\mu{_k}(\t f{_j}) \,|\t \alpha{_i_j}|   \, \| \t f{_j} \|
			\right)^2,
	\end{align}
	where $\t \alpha{_i_j}$ are complex coefficients of the basis transformation $\t*f{^\prime_i} = \sum_j \t\alpha{_i_j} \t f{_j}$.
	Since the right-hand side is finite by assumption, all $\t*f{^\prime_i}$ have finite $k$-moments.
	However,  we then have that \cref{eq:moment transformation eigenvectors}, with $\t f{_i}$ replaced by $\t*f{^\prime_i}$, implies that $\chi = 1$.
\end{proof}

Since the divergence of all moments of the spectral distribution functions is physically implausible (e.g., for Gaussian distributions, which are commonly employed in many areas of physics, all moments are finite), considering typical scenarios of interest will inevitably imply the absence of finite-dimensional invariant subspaces of $\t U{_\chi}$ for $\chi \neq 1$.
As a consequence, the \MMM\ method for constructing unitary matrices from $\t U{_\chi}$ is generally not feasible.
In view of these considerations, it thus comes at no surprise that the matrices computed in the original work fail to be unitary~\cite{Leber_2026}.

\section{Resolution: Compressions and dilations}
\label{s:compressions and dilations}

Having established that the \MMM\ model, in its original formulation, is inconsistent, we now present two consistent alternatives:
\vspace{\baselineskip}

\noindent \emph{Self-contained approach}: The first approach is based on the observation that the observables considered in the original literature \cite{2023AnP...53500468B,2023JPhCS2531a2016A} can be computed solely using the matrix $\t u{_\chi} = (\t f{_i}, \t U{_\chi} \t f{_i})$, with $i, j = 1, \ldots, n$, whose non-unitarity implies no inconsistency of the model.

\vspace{\baselineskip}

\noindent \emph{Extension method}: The second approach is to extend $u{_\chi}$ to a unitary transformation $\t\Upsilon{_\chi}$ on a larger space by adding “auxiliary” degrees of freedom. The relevant observables for the field modes of interest can then be computed using reduced quantum states.

\vspace{\baselineskip}

As is explained in more detail below, the first approach arises naturally from the field-theoretic setting without \emph{ad-hoc} constructions.
The second approach, on the other hand, is conceptually closer to the original work on the \MMM\ model, as it introduces auxiliary degrees of freedom to describe redshifts using unitary matrices, which allows for the application of standard quantum-information-theoretic techniques (an example is provided in \cref{s:quantum channel}).
Both approaches work solely using finite-dimensional matrices and thus avoid the mathematical difficulties associated to working with infinite-dimensional matrices that were encountered previously~\cite{Leber_2026}.

\subsection{Compressions}
\label{s:compressions}

From a mathematical point of view, if $V = \linearHull \{\t f{_1}, \ldots, \t f{_n}\}$ is not closed under the action of $\t U{_\chi}$, it follows that the coefficients
\begin{align}
    \label{eq:compression coefficients}
    \t u{_i_j}(\chi)
		= (\t f{_i}, \t U{_\chi} \t f{_j})
\end{align}
cannot be regarded as the matrix elements of the \emph{restriction} of $\t U{_\chi}$ to $V = \linearHull \{\t f{_1}, \ldots, \t f{_n}\}$ precisely because $\t U{_\chi} \big|_V$ does not have $V$ as its codomain.
Instead, \cref{eq:compression coefficients} describes the coefficients of the operator
\begin{align}
    \label{eq:compression abstract}
    \t u{_\chi} = \t P{_V} \t U{_\chi} \big|_V,
\end{align}
where $\t P{_V}$ is the orthogonal projector onto $V$. As such, $\t u{_\chi}$ is a finite-dimensional \emph{compression} of $\t U{_\chi}$ rather than a \emph{restriction} \cite[Chap.~23]{Halmos:1982}, and our theorem above implies that $\t u{_\chi}$ is generically non-unitary.
As we will now show, this compression naturally arises in field-theoretic computations of single-photon transmission amplitudes.

Indeed, standard Fock quantization techniques yield ladder operators $\hat c$ and $\hat c^\dagger$ that are associated to Killing spectra \cite{1975RSPSA.346..375A,Fulling:1989}.
\Cref{eq:spectrum transformation} can then be used to define ladder operators $\hat a, \hat a^\dagger$ and $\hat b, \hat b^\dagger$ that are associated to stationary observers at positions $\t p{_\!_\alice}$ and $\t p{_\!_\bob}$, namely
\begin{align}
	\hat c(\t f{_*}) &= \hat a(\t f{_\alice}) = \hat b(\t f{_\bob}),
\end{align}
where
\begin{align}
	\t*f{_\AB}(\omega) &= \sqrt{\t\lapse{_\AB}}\, \t f{_*}(\t\lapse{_\AB} \,\omega ).
\end{align}
From \cref{eq:spectrum bob from alice,eq:U operator definition} it then follows that
\begin{align}
	\label{eq:ladder operator transformation}
	\hat a(f) = \hat b(\t U{_\chi} f),
\end{align}
where $\t U{_\chi}$ is defined in \cref{eq:U operator definition}, and $\chi = \sqrt{\t\lapse{_\alice} / \t\lapse{_\bob}}$ is the redshift parameter.
The gravitational frequency shift of wave packets can thus be described using \cref{eq:ladder operator transformation} without any reference to Killing spectra and also without a finite-dimensional matrix representation of $\t U{_\chi}$.
Crucially, when considering single-mode or multi-mode transmission protocols as described in the literature~\cite{2023AnP...53500468B,2023JPhCS2531a2016A}, the map \eqref{eq:compression abstract} and its coefficients \eqref{eq:compression coefficients} arise naturally, as we demonstrate in the following paragraphs. 

\emph{Single-Mode Transmission}---The probability amplitude of detecting a photon with frequency spectrum $f$ in the frame of observer $B$ if a photon with frequency spectrum $f$ was emitted by $A$ is given by
\begin{multline}
	\braket{0|\hat b(f) \hat a^\dagger(f)|0}
		= \braket{0|\hat b(f) \hat b^\dagger(\t U{_\chi} f)|0}
        \\
		= \braket{0|[\hat b(f), \hat b^\dagger(\t U{_\chi} f)]|0}
		= (f, \t U{_\chi} f)\equiv\Delta_f(\chi).
\end{multline}
Here we have introduced again $\Delta_f(\chi)$, which was defined previously in \cref{delta:definition}.
Of course, this equation can be derived using \cref{eq:redshift f h decomposition} with the supplementary spectrum $h$ defined in \cref{eq:h definition}:
\begin{align}
\begin{split}
    \label{eq:ladder operator relation}
	\hat a^\dagger(f)
		&= \hat b^\dagger(\t U{_\chi} f)
        \\&
		= (f, \t U{_\chi} f) \hat b^\dagger(f) + (h, \t U{_\chi} f) \hat b^\dagger(h).
\end{split}
\end{align}
However, the straight-forward calculation above shows that a supplementary orthogonal mode is not needed for making explicit predictions.
Hence, the single number $\Delta_f(\chi)=(f, \t U{_\chi} f)$ suffices to characterize the degree of distinguishability of the shifted spectrum $\t U{_\chi} f$ and the original spectrum $f$. In other words, a $2\times 2$-matrix is not needed.

\emph{Multi-Mode Transmission}---If the spectrum relative to $A$ is chosen from a predetermined collection  $(\t f{_i})_{i = 1}^n$ of mode functions, the probability amplitude for passing a filter from the same collection at $B$ is given by
\begin{multline}
    \label{eq:multi-mode transmission amplitudes}
	\braket{0|\hat b(\t f{_j}) \hat a^\dagger(\t f{_i})|0}
		= \braket{0|\hat b(\t f{_j}) \hat b^\dagger(\t U{_\chi} \t f{_i})|0}\\
		= \braket{0|[\hat b(\t f{_j}), \hat b^\dagger(\t U{_\chi} \t f{_i})]|0}
		= (\t f{_j}, \t U{_\chi} \t f{_i})
		.
\end{multline}
Hence, all relevant information is summarized in the coefficients defined in \cref{eq:compression coefficients}, or equivalently in the operator $\t u{_\chi}$ defined in \cref{eq:compression abstract}.
The probability of not matching any of the $\t f{_j}$ is thus given by $p(i) = 1 - \sum{_j} |\t u{_j_i}(\chi)|^2$.
Of course, this can be expressed as
\begin{align}
	p(i) = \braket{0|\hat b(\t k{_i}) \hat b^\dagger(\t f{_i})|0}
\end{align}
where $\t k{_i}$ is the perpendicular mode
\begin{align}
	\t k{_i}
		= 
	(\t U{_\chi} \t f{_i})^\perp
		= \t U{_\chi} \t f{_i} - \sum_{j = 1}^n (\t f{_j} , \t U{_\chi} \t f{_i})\, \t f{_j},
\end{align}
provided that the $\t f{_j}$ are orthonormal.
However, it should be clear from the discussion above that there is logically no need to extend the set $\t f{_i}$ by $\t k{_i}$, and hence to consider not only $\t U{_\chi} \t f{_i}$ but also $\t U{_\chi} \t k{_i}$, or even additional modes orthogonal to this extended collection of spectra.

These examples show that single-mode and multi-mode transmission protocols put forward in previous work~\cite{2023AnP...53500468B,2023JPhCS2531a2016A} can be understood entirely in terms of the compression \eqref{eq:compression abstract} with its matrix components \eqref{eq:compression coefficients}, whose non-unitarity simply comes from the fact that no finite collection of frequency spectra (with finite moments) is invariant under redshifts.

These examples also illustrate a general point concerning the original \MMM\ model:
the matrices considered there have the block structure
\begin{align}
    \label{eq:MMM block matrix}
    U_\MMM
    = 
    \begin{pmatrix*}
        \t u{_\chi} & *\\
        * & *
    \end{pmatrix*}.
\end{align}
The transmission amplitudes involving only “main” modes depend solely on the upper left block, $\t u{_\chi}$, hence the \MMM\ predictions for observables involving no auxiliary modes are consistent with quantum field theory in stationary space-times \emph{regardless} of the size of the frequency shift.
The problem of the model consists in its incorrect predictions concerning the auxiliary modes: if all matrix entries are chosen to be $(\t f{_I}, \t U{_\chi} \t f{_J})$, with $\t f{_I}$ and $\t f{_J}$ ranging over main and auxiliary spectra, $U_\MMM$ is generically not unitary (as proven above), regardless of the number of auxiliary modes.

\subsection{Dilations}
\label{s:dilations}

\Cref{eq:MMM block matrix} suggests a possible way of correcting the \MMM\ model by filling the remaining blocks in such a way as to obtain a unitary matrix. As described above, this comes at the cost of having to fill the matrix using an \emph{ad-hoc} selection procedure since the “natural” choice $(\t f{_I}, \t U{_\chi} \t f{_J})$ leads to a non-unitary matrix.
In the mathematical literature, such an “enlarged” operator is known as a \emph{unitary dilation}, i.e., a unitary operator $\t\Upsilon{_\chi}$ on a space $W$ with an embedding $V \hookrightarrow W$ such that $\t u{_\chi}$ is a compression of $\t\Upsilon{_\chi}$.
The recent article \cite{Leber_2026} proved that \emph{if} $\t u{_\chi}$ admits a unitary dilation to a space of dimension $n +n'$, then $n' \geq n$, but the existence of such a dilation was not proven there.
However, a result commonly attributed to Halmos shows that unitary dilations of $\t u{_\chi}$ always exist~\cite{Robinson:2018}: one may choose $W = V \oplus V$ and set
\begin{align}\label{redshift:matrix:dilation}
    \t\Upsilon{_\chi} =
        \begin{pmatrix}
            \t u{_\chi} & - \sqrt{\identity - \t*u{_\chi} \t*u{_\chi^\dagger}}
            \\
            + \sqrt{\identity - \t*u{_\chi^\dagger} \t*u{_\chi}} & \t*u{_\chi^\dagger}
        \end{pmatrix}.
\end{align}
Such a dilation provides an accurate representation of the inner products of the spectra $\t f{_i}$ with their shifted versions $\t U{_\chi} \t f{_i}$ that maintains unitarity by deviating from $\t U{_\chi}$ in all matrix components that are associated to supplementary spectra $\t h{_i}$.
This construction can be considered as a minimal correction to the \MMM\ model, which manifestly agrees with its original predictions for the transition amplitudes $(\t f{_i}, \t U{_\chi} \t f{_j})$ as discussed above.

Since many quantum-information-theoretic techniques require unitary transformations, the operator $\t\Upsilon{_\chi}$ makes it possible to use such methods to analyze the gravitational redshift in more detail. An example of such an analysis is described below.

\section{Gravitational redshift as a Gaussian quantum channel}
\label{s:quantum channel}

We now proceed to demonstrate how the unitary operator $\t\Upsilon{_\chi}$ can be used to model the gravitational redshift of Gaussian states as a quantum channel.
To this end, we first introduce a Hamiltonian $\hat H$ on a suitable Fock space such that the corresponding unitary evolution operator $\hat U$ correctly reproduces the multi-mode transmission amplitudes discussed in \cref{s:compressions and dilations}.
When applied to Gaussian states, it is trivial to “trace out” the auxiliary degrees of freedom associated to the dilation. As we demonstrate explicitly, the induced transformation on the “main” subsystem takes the form of a Gaussian quantum channel.
A graphical overview of the technical steps involved in this construction is provided in \cref{fig:flowchart}.

\begin{figure}[t]
    \centering
    \includegraphics{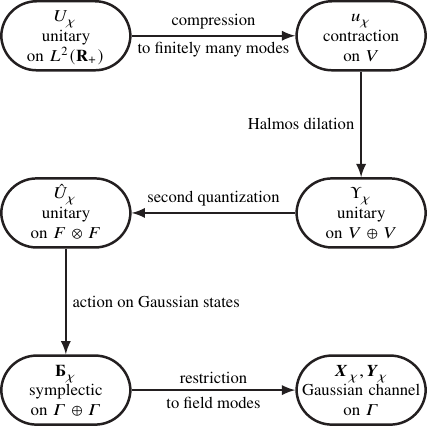}
    \caption{%
    Visual representation of the construction of the Gaussian quantum channel associated to the gravitational frequency shift $\t U{_\chi}$. By considering transition amplitudes among a finite number of modes, one obtains a non-unitary contraction $\t u{_\chi}$ on $V$, the span of the considered field modes. \Cref{redshift:matrix:dilation} extends this to a unitary operator $\t \Upsilon{_\chi}$ on the enlarged space $V \oplus V$. \Cref{eq:second quantization:hamiltonian,eq:second quantization:unitary} describe the corresponding operator $\t{\hat U}{_\chi}$ on the Fock space $F(V \oplus V) \simeq F(V) \otimes F(V)$. The induced action on Gaussian states gives a symplectic transformation $\t\Bogolyubov{_\chi}$ on the enlarged phase space $\varGamma \oplus \varGamma$, whose restriction is the quantum channel described in \cref{eq:Gaussian quantum channel result}.}
    \label{fig:flowchart}
\end{figure}

\subsection{Operators on Fock space}

We start by denoting $F(H)$ as the symmetric Fock space over a Hilbert space $H$,
and we consider $H$ to be the vector space $V$ spanned by $n$ field modes with frequency spectra $\t f{_i}$, endowed with the Klein–Gordon inner product.
The unitary operator \eqref{redshift:matrix:dilation} acts on $V \oplus V$, a space of dimension $N = 2n$, the associated Fock space of which is $F(V \oplus V) \simeq F(V) \otimes F(V)$.
Denoting by $\t\ha{_1}, \ldots, \t\ha{_{2n}}$ the associated ladder operators (the first $n$ of which correspond to the field’s degrees of freedom, the remaining $n$ to the auxiliary degrees of freedom), the Hamiltonian associated to $\t\Upsilon{_\chi}$ can be constructed as follows.
We use the matrix logarithm and set $\t h{_\chi} = \ii \ln \t\Upsilon{_\chi}$, which is a Hermitian matrix of dimension $N \times N = 2n \times 2n$, and the corresponding second-quantized operator is
\begin{align}
    \label{eq:second quantization:hamiltonian}
    \hat H(\chi) = \sum_{I = 1}^{N} \sum_{J = 1}^{N} \t h{_I_J}(\chi) \t*{\ha}{_I^\dagger} \t*{\ha}{_J}.
\end{align}
This construction is such that the unitary operator
\begin{align}
    \label{eq:second quantization:unitary}
    \hat U(\chi) = \exp(-\ii \hat H(\chi))
\end{align}
leaves the vacuum state invariant, $\hat U(\chi) \ket 0 = \ket 0$, since $\hat H(\chi)$ is normal-ordered, and transforms the ladder operators in the following manner
\begin{subequations}
\begin{align}
    \label{eq:unitary trf:lowering}
    \hat U^\dagger(\chi) \t*{\ha}{_I} \hat U(\chi)
        &= \sum_{J = 1}^{N} \t*\Upsilon{_I_J}(\chi) \t*{\ha}{_J},
    \\
    \label{eq:unitary trf:raising}
    \hat U^\dagger(\chi) \t*{\ha}{^\dagger_I} \hat U(\chi)
        &= \sum_{J = 1}^{N} \t*\Upsilon{_I_J^*}(\chi) \t*{\ha}{^\dagger_J}.
\end{align}
\end{subequations}
This construction gives an accurate model of the gravitational redshift of photons with frequency spectra $\t f{_1}, \ldots, \t f{_n}$,  as can be verified by computing the transition amplitudes
\begin{multline}
    \braket{0 | \t{\ha}{_i} \, \hat U(\chi) \t*{\ha}{^\dagger_j} | 0}
    = \braket{0 | \hat U^\dagger(\chi) \, \t{\ha}{_i} \, \hat U(\chi) \t*{\ha}{^\dagger_j} | 0}
    \\
    = \t\Upsilon{_i_j}(\chi)
    = \t u{_i_j}(\chi)
    = (\t f{_i}, \t U{_\chi} \t f{_j}),
\end{multline}
which reproduce the purely field-theoretic result \eqref{eq:multi-mode transmission amplitudes}.

\subsection{Gaussian quantum channels}
\label{s:Gaussian quantum channel}

Any Gaussian state on the Fock space over $V \oplus V$ is fully characterized by its first and second moments
\begin{align}
    \t d{_A} &= \braket{\t\Y{_A}},
    &
    \t\Sigma{_A_B} &= \braket{ \{\t*\Y{_A}, \t*\Y{_B^\dagger}\} } - 2 \braket{ \t*\Y{_A} } \braket{\t*\Y{_B^\dagger} },
\end{align}
where $\{\cdot,\cdot\}$ denotes the anti-commutator and $\Y = (\t{\ha}{_1}, \ldots, \t{\ha}{_N}, \t*{\ha}{_1^\dagger}, \ldots, \t*{\ha}{_N^\dagger})$, see, e.g., Ref.~\cite{Adesso:Ragy:2014}.
These expectation values can be written in matrix form as
\begin{align}
    \label{eq:gaussian d Sigma decomposition}
    \bs d &= \begin{pmatrix*}[l]
        \alpha\\
        \alpha^*
    \end{pmatrix*},
    &
    \bs\Sigma &= \begin{pmatrix*}[l]
        \sigma & \tau\\
        \tau^* & \sigma^*
    \end{pmatrix*},
\end{align}
where $\sigma$ is self-adjoint and $\tau$ is symmetric (both are matrices of size $N \times N$), and asterisks denote complex conjugation. In this notation, the reduced state of a subsystem is obtained simply by omitting the rows and columns associated to its complement.

As explained in more detail below, the unitary operator $\t\Upsilon{_\chi}$ introduced in \cref{redshift:matrix:dilation} naturally induces a transformation of the quantities $\bs d$ and $\bs \Sigma$ of the \emph{whole enlarged system}, which in turn allows us to characterize the process as a quantum channel acting on the reduced system corresponding to the original field modes of interest.

If $\varrho$ denotes the density matrix of a Gaussian state, the moments of the transformed state $\hat U(\chi) \varrho \hat U(\chi)^\dagger$ are
\begin{align}
    \label{eq:transformation alpha sigma tau}
    \alpha_\chi
        &= \t*\Upsilon{_\chi} \alpha_1,
    &
    \sigma_\chi
        &= \t*\Upsilon{_\chi} \t{\sigma}{_1} \t*\Upsilon{^\dagger_\chi},
    &
    \tau_\chi
        &= \t*\Upsilon{_\chi} \t{\tau}{_1} \t*\Upsilon{^\transpose_\chi}.
\end{align}
These equations can alternatively be expressed as
\begin{align}
    \label{eq:unitary transformation d Sigma}
    \bs d_\chi
        &= \t*{\bs\Bogolyubov}{_\chi} \bs d_1,
    &
    \bs \Sigma_\chi
        &= \t*{\bs\Bogolyubov}{_\chi} \bs \Sigma_1 \t*{\bs\Bogolyubov}{_\chi^\dagger},
\end{align}
where $\bs\Bogolyubov_\chi$ is the block-diagonal Bogolyubov transformation
\begin{align}
    \t*{\bs\Bogolyubov}{_\chi}
        = \begin{pmatrix*}[l]
            \Upsilon_\chi & 0\\
            0 & \Upsilon^*_\chi
        \end{pmatrix*}.
\end{align}
The absence of off-diagonal entries in $\t*{\bs\Bogolyubov}{_\chi}$ is consistent with the fact the operator $\t U{_\chi}$ defined in \cref{eq:U operator definition} does not mix positive and negative frequencies, or equivalently that \cref{eq:unitary trf:lowering,eq:unitary trf:raising} do not mix raising and lowering operators (i.e., no particle creation).

Now, because the modes $1, \ldots n$ describe the actual field modes while the modes $n + 1, \ldots, 2n$ are auxiliary, it is natural to decompose the matrices $\sigma_\chi$ and $\tau_\chi$ in block form as
\begin{align}
    \sigma
        &= \begin{pmatrix*}[l]
                \t*\sigma{_\sys}
            &   \t*\sigma{_\corr}
            \\
                \t*\sigma{_\corr^\dagger}
            &   \t*\sigma{_\aux}
        \end{pmatrix*},
    &
    \tau
        &= \begin{pmatrix*}[l]
                \t*\tau{_\sys}
            &   \t*\tau{_\corr}
            \\
                \t*\tau{_\corr^\dagger}
            &   \t*\tau{_\aux}
        \end{pmatrix*}.
\end{align}
Just as in the single-photon case described above, it is natural to consider initial states without auxiliary excitations, hence
\begin{subequations}
\begin{align}
    \t*\sigma{_\corr}(\chi = 1) &= 0,
    &
    \t*\tau{_\corr}(\chi = 1) &= 0,
    \\
    \t*\sigma{_\aux}(\chi = 1) &= \identity,
    &
    \t*\tau{_\aux}(\chi = 1) &= 0.
\end{align}
\end{subequations}
Since the auxiliary degrees are considered unobservable, it suffices to compute the final values of $\sigma_\sys$ and $\tau_\sys$. \Cref{eq:transformation alpha sigma tau} yields
\begin{subequations}
\begin{align}
    \t{\sigma}{_\sys}(\chi)
        &= \t*u{_\chi} \t{\sigma}{_\sys}(1) \t*u{_\chi^\dagger} 
         + \identity - \t*u{_\chi} \t*u{_\chi^\dagger},
    \\
    \t{\tau}{_\sys}(\chi)
        &= \t*u{_\chi} \t{\tau}{_\sys}(1) \t*u{_\chi^\transpose}.
\end{align}
\end{subequations}
When considering the field modes only, the gravitational frequency shift can thus be regarded as a Gaussian quantum channel $\t\channel{_\chi}$, see \cite{2005quant.ph..5151E}, of the form
\begin{align}
    \label{eq:Gaussian quantum channel result}
    \t\channel{_\chi}(\bf\Sigma_\sys)
        = \t*{\bs X}{_\chi} \bf\Sigma_\sys \t*{\bs X}{_\chi^\dagger} + \t*{\bs Y}{_\!_\chi},
\end{align}
where $\bs\Sigma_\sys$ is the $2n\times2n$ covariance matrix of the field modes with general form
\begin{align}
    \label{eq:covariance matrix decomposition}
    \bs\Sigma_\sys = \begin{pmatrix}
        \sigma_\sys & \tau_\sys
        \\
        \tau_\sys^* & \sigma_\sys^*
    \end{pmatrix},
\end{align}
and the $2n\times2n$ matrices $\bs X$ and $\bs Y$ take the form
\begin{align}
    \t{\bs X}{_\chi} &= \begin{pmatrix}
            \t*u{_\chi}
        &   0
        \\
            0
        &    \t*u{_\chi^*}
    \end{pmatrix},
    &
    \t{\bs Y}{_\!_\chi} &= \begin{pmatrix}
            \identity - \t*u{_\chi} \t*u{_\chi^\dagger} 
        &   0
        \\
            0
        &   \identity - \t*u{_\chi^*} \t*u{_\chi^\transpose}
    \end{pmatrix}.
\end{align}
This result corroborates the fact that the whole evolution of the reduced state of the modes of interest is determined uniquely by the matrix $u_\chi$ of measurable coefficients $\t u{_i_j}(\chi) = ( \t f{_i}, \t U{_\chi} \t f{_j})$. Thus, in any practical scheme no knowledge of the auxiliary modes $\t h{_j}$ is needed.

\subsection{Quantifying non-unitarity of the redshift-induced Gaussian quantum channel}
\label{s:Gaussian quantum channel:non-unitarity:quantification}

We now wish to find a way to quantify the deviation from unitarity of $\t u{_\chi}$. Below we provide two complementary methods to achieving this goal.

\textit{Fidelity measures}---A standard way of comparing two quantum states with density operators $\hat\rho$ and $\hat\rho'$ is to employ the fidelity $\fidelity$~\cite{Nielsen:Chuang:2000}, defined as
\begin{align}
    \label{quantum:state:fidelity}
    \mathcal{F}(\hat{\rho},\hat{\rho}')
        =
        \left(\trace\left[
            \sqrt{\sqrt{\hat\rho} \, \hat\rho' \, \sqrt{\hat\rho}}
        \right]\right)^2.
\end{align}
For pure states $\hat\rho = \ketbra{\psi}{\psi}$ and $\hat\rho'= \ketbra{\psi'}{\psi'}$ the fidelity reduces to $\mathcal F(\hat\rho, \hat\rho') = |\braket{\psi \mid \psi'}|^2$.

In the particularly simple case of two-mode Gaussian states with vanishing first moments, the fidelity can be expressed in terms of the respective covariance matrices $\bs\Sigma$ and $\bs\Sigma'$ \cite{Paraoanu:Scutaru:2000,Marian:Marian:2012}. Explicitly, one has
\begin{align}
    \fidelity(\bs\Sigma, \bs\Sigma')
        = \frac{4}{ \sqrt\Lambda + \sqrt\Gamma-\sqrt{(\sqrt\Lambda + \sqrt\Gamma)^2 -\Delta}},
\end{align}
where
\begin{subequations}
\begin{align}
    \Gamma &= \det(\identity_4 - \bs\Omega \, \bs\Sigma \, \bs\Omega \, \bs\Sigma'),
    \\
    \Lambda &=\det(\identity_4 +\ii \,\bs\Omega \, \bs\Sigma )  \det(\identity_4 + \ii\,\bs\Omega \, \bs\Sigma'),
    \\
    \Delta &= \det(\bs\Sigma + \bs\Sigma'),
    \\
    \intertext{in which}
    \bs\Omega
        &= \operatorname{diag}(
            \underbrace{-\ii, \ldots, -\ii}_{\text{$N = 2n$ times}},
            \underbrace{+\ii, \ldots, +\ii}_{\text{$N = 2n$ times}}
        ).
\end{align}
\end{subequations}
It is immediate to verify that $\Lambda = 0$ when either $\hat\rho$ or $\hat\rho'$ is pure \cite{Heib:Lageyre:2025}.
Similarly, using Williamson’s theorem \cite{Williamson:1936}, one can demonstrate that purity of either one of the states implies $\Gamma = \Delta$ (see, e.g., App.~B of Ref.~\cite{Heib:Lageyre:2025}).
% Using the same arguments presented in the literature, it is immediate to verify that, when $\boldsymbol{\sigma}=\boldsymbol{s}_0\boldsymbol{s}_0^\dag$ is pure and uniquely determined by the symplectic matrix $\boldsymbol{s}_0$, one has
% \begin{align}
%     \begin{split}
%     \Delta=&\det(\boldsymbol{\sigma}+\boldsymbol{\sigma}')\\
%     =&\det(\boldsymbol{s}_0\boldsymbol{s}_0^\dag+\boldsymbol{\sigma}')\\
%     =&\det(\mathds{1}+\boldsymbol{s}_0^\dag{}^{-1}\boldsymbol{s}_0{}^{-1}\boldsymbol{\sigma}')\\
%     =&\det(\mathds{1}-\boldsymbol{\Omega}\boldsymbol{s}_0\boldsymbol{s}_0^\dag\boldsymbol{\Omega}\boldsymbol{\sigma}')\\
%     =&\det(\mathds{1}-\boldsymbol{\Omega}\boldsymbol{\sigma}\boldsymbol{\Omega}\boldsymbol{\sigma}')=\Gamma.
%     \end{split}
% \end{align}
% We conclude that it is sufficient that one of the states is pure to have $\Delta=0$ and $\Lambda=\Gamma$. 	
Thus, in such a case the fidelity takes the simple form
\begin{align}
    \fidelity(\bs\Sigma, \bs\Sigma')
        = \frac{4}{\sqrt{\Delta}}
        = \frac{4}{\sqrt{\det(\bs\Sigma + \bs\Sigma')}}.
\end{align}
Hence, the action of gravitational redshift on an initially pure state with covariance matrix $\bs\Sigma$ yields a final state with covariance matrix $\t\channel{_\chi}(\bs \Sigma)$ whose fidelity relative to the initial state, $\t\fidelity{_\chi}(\bs\Sigma) := \fidelity(\bs\Sigma, \t\channel{_\chi}(\bs\Sigma))$, takes the form
\begin{align}
    \label{eq:fidelity redshifted pure state}
    \t\fidelity{_\chi}(\bs\Sigma)
        = \frac{4}{\sqrt{\det( \bs\Sigma + \t*{\bs X}{_\chi} \bs\Sigma \t*{\bs X}{_\chi^\dagger} + \t*{\bs Y}{_\chi})}}.
\end{align}
If the initial state is pure, the Bloch–Messiah theorem implies that its covariance matrix can be written as ${\bs\Sigma = (\bs V \bs \Xi) (\bs V \bs \Xi)^\dagger}$, where $\bs V$ is unitary and $\bs \Xi$ is a Hermitian squeezing matrix \cite{Braunstein:2005}. Using this factorization, the fidelity can be written as
\begin{align}
    \label{eq:fidelity redshift pure squeezing}
    \t\fidelity{_\chi}(\bs\Sigma)
        = \frac{4}{\sqrt{\det( \bs\identity + \bs\Xi^2 + \bs{\tilde X}(\bs\Xi^2 - \bs\identity) \bs{\tilde X}^\dagger) }},
\end{align}
where $\t{\bs{\tilde X}}{_\chi} = \bs V \t{\bs X}{_\chi} \bs V^\dagger$.
Hence, in the absence of squeezing, $\bs \Xi = \bs \identity$, one has $\t\fidelity{_\chi} = 1$ regardless of the redshift parameter $\chi$, and therefore lower values $\t\fidelity{_\chi} < 1$ can be obtained only through squeezing.
This establishes squeezing as a paramount resource for measuring fidelity reductions induced by gravitational frequency shifts.

\textit{Quantification of the noise term}---An alternative, state-independent, measure for the effective non-unitarity in this setup is provided by the normalized trace $\eta(\chi)
        := \frac{1}{2 n} \trace \t{\bs Y}{_\!_\chi}$ of the matrix $\t{\bs Y}{_\!_\chi}$, which we call \emph{channel nonunitarity}, that explicitly reads
\begin{equation}
    \begin{split}
    \eta(\chi)
        &= 1 - \frac{1}{n} \sum_{i, j = 1}^n |\t u{_i_j}(\chi)|^2
        \\
        &= 1 - \frac{1}{n} \sum_{i, j = 1}^n |(\t f{_i}, \t U{_\chi} \t f{_j})|^2.
    \end{split}
\end{equation}

\noindent The Cauchy–Schwarz inequality implies that $\eta(\chi) \geq 0$, and it is easy to see that the equality $\eta(\chi) = 0$ is equivalent to $\t u{_\chi}$ being unitary.
By virtue of our main theorem, however, this occurs in realistic scenarios if and only if $\chi = 1$, so generically non-trivial frequency shifts imply $\eta(\chi) \neq 0$.

Crucially, due to the properties of the transformation $U_\chi$, we have $\eta(\chi) = \eta(1/\chi)$, which we write as
\begin{equation}\label{up:down:relation}
    \eta_{\uparrow}(\chi)=\eta_{\downarrow}(\chi).
\end{equation}
This holds for any pair of uplink and downlink processes (labelled with an upward and downward pointing arrow respectively).

\subsection{Roundtrip quantum communication protocol}

Besides comparing quantum states relative to different observers $\alice$ and $\bob$, it is also possible to characterize the quantum channel associated to the round-trip $\alice \to \bob \to \alice$.
Crucially, the state received at the end at $\alice$ depends on the details of the reflection at $\bob$. Clearly, a reflection by an “ideal mirror” at $\bob$ merely reverses the direction of propagation but leaves the frequency spectrum invariant, so that the overall transformation is $\t U{_{1/\chi}} \t U{_\chi} = \identity$ on account of \cref{eq:U composition}. However, if the reflection at $\bob$ occurs at an “imperfect” mirror that reflects only part of the spectrum, the state received at $\alice$ will differ from the one emitted at the same spot. 

As a simple model, one may consider the case where perfect reflection occurs if the frequency spectrum at $\bob$ is within the span of some finite set of frequency spectra $\t f{_i}$, and orthogonal spectra are fully absorbed. 
Since $\t u{_{1/\chi}} = \t*u{_\chi^\dagger}$, the Gaussian quantum channel for such a “lossy” round-trip $\alice \to \bob \to \alice$ takes the form
\begin{align}
    \t{\bs \Sigma}{_\sys}
        \mapsto \t*{\bs Z}{_\chi} \t{\bs \Sigma}{_\sys} \t*{\bs Z}{_\chi^\dagger}
        + \bs\identity - \t*{\bs Z}{_\chi} \t*{\bs Z}{_\chi^\dagger},
\end{align}
where
\begin{align}
    \t*{\bs Z}{_\chi}
        = \t*{\bs Z}{_\chi^\dagger}
        = \t*{\bs X}{_\chi^\dagger} \t*{\bs X}{_\chi}
        = \begin{pmatrix}
            \t*u{_\chi^\dagger} \t*u{_\chi} & 0\\
            0 & \t*u{_\chi^\transpose} \t*u{_\chi^*}
        \end{pmatrix}.
\end{align}
Evidently, this differs from the identity transformation whenever $\chi \neq 1$ because $\t u{_\chi}$ cannot be unitary for non-trivial frequency shifts. A pictorial illustration of this scheme is given in Figure~\ref{fig:schematic ABA}.

\begin{figure}
    \centering
    \includegraphics[width=\linewidth]{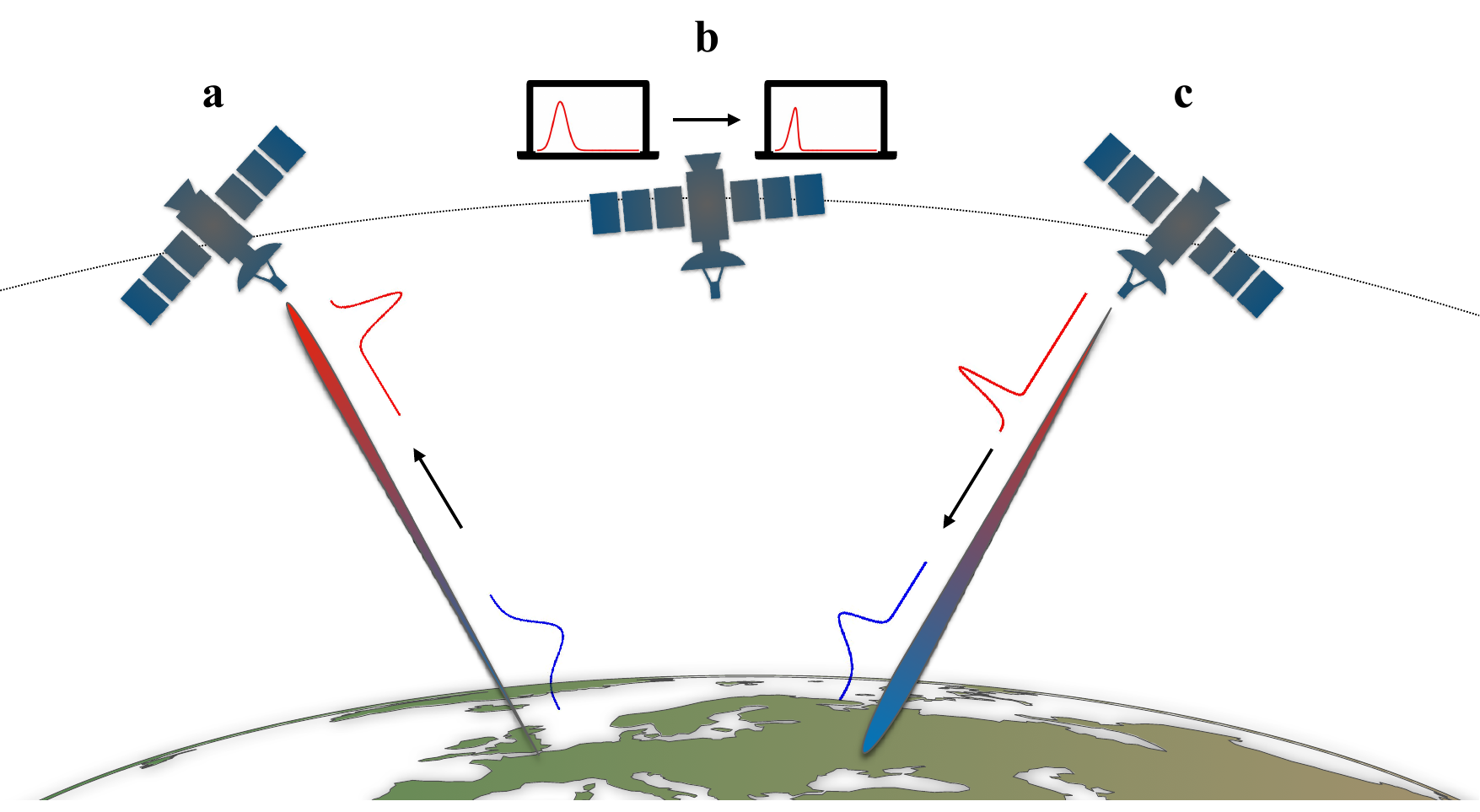}
    \caption{
        \textbf{Round-trip protocol}: A photon is sent from site \alice\ to site \bob, where the received frequency spectrum deviates from the one emitted by \alice\ (Panel a).
        At site \bob, the wave-packet’s frequency spectrum is modified by means of imperfect reflection or an appropriate projective measurement (Panel b).
        If the transformed wave packet is sent back to \alice, the received photon state differs from the one emitted before (Panel c).
        Note that the present mathematical analysis assumes stationary observers and thus applies (approximately) to geostationary satellites, but not to other satellite orbits that move relative to a ground station.
    }\label{fig:schematic ABA}
\end{figure}

We can quantify the deviation from unitarity in this simple roundtrip scheme by using the channel nonunitarity $\eta$. We find that the roundtrip channel nonunitarity $\eta_\updownarrow$ reads
\begin{equation}
    \eta_\updownarrow(\chi)
        = 1 - \frac{1}{n} \sum_{i, j,k = 1}^n |\t*u{_i_j^*}(\chi) \t*u{_j_k}(\chi)|^2.
\end{equation}
By expanding the sum in an eigenbasis of $\t*u{_\chi^\dagger} \t u{_\chi}$, one finds
\begin{align}
    \eta_\updownarrow(\chi)
        \geq \eta_\uparrow(\chi)
        \geq \eta_\uparrow(\chi)^2,
\end{align}
where $\eta_\uparrow(\chi) = 1 - \frac{1}{n} \sum_{i, j = 1}^n |\t u{_i_j}(\chi)|^2$ is the non-unitarity of the one-way uplink. An identical equation applies for the downlink: this follows directly from \eqref{up:down:relation}. Since the channel nonunitarity is always larger than zero, this implies that a round-trip process is “more nonunitary” than the one-way process when a filter is applied at the middle node.

\section{Considerations}
\label{s:MMM review of results}

The observation that $\t u{_i_j}(\chi)$ is not unitary for $\chi \neq 1$ prompts the question whether strong redshifts preserve the photon number relative to stationary observers~\cite{Leber_2026}.
The present analysis answers this affirmatively: \cref{eq:ladder operator relation} shows that if an observer \alice\ creates the single-photon state $\ha^\dagger(f) \ket 0$, the same state written in a basis adapted to another stationary observer \bob\ is given by $\hat b^\dagger(\t U{_\chi} f) \ket 0$, where $\t U{_\chi}$ is the operator defined in \cref{eq:U operator definition} (the present analysis uses the Heisenberg picture, in which state vectors \emph{do not} evolve).
The transformation $\t U{_\chi}$ of complex frequency spectra is unitary on the whole infinite-dimensional space of mode functions, however, computing “matrix elements” $\t u{_i_j}(\chi)$ for a finite collection of modes generically leads to non-unitary matrices because the frequency-shift transformation $\t U{_\chi}$ has no regular finite-dimensional invariant subspaces.
Nevertheless, the non-unitarity of the resulting matrix representation does not imply any photon loss, let alone an inconsistency of the model: it is a logical consequence of the explicit form of the frequency-shift operator $\t U{_\chi}$.

The present analysis partially confirms and partially corrects results obtained in the previous articles \cite{2023AnP...53500468B,2023JPhCS2531a2016A,Leber_2026}.
Specifically, given that all previous work agree on the upper left block of the matrix \eqref{eq:MMM block matrix}, this work confirms the prior results on single-photon transmission amplitudes and probabilities involving field modes.
Since the auxiliary modes introduced in prior analyses were assumed experimentally inaccessible, the remaining entries of the matrices considered there do not lead to experimentally falsifiable predictions for single-photon transmission. Furthermore, the prior analyses were restricted to weak gravitational fields, while our analysis extends these results to \emph{arbitrarily strong} gravitational frequency shifts.

We emphasize that it was previously assumed to be sufficient to add \emph{a single} auxiliary mode to obtain a unitary matrix, from which one could then construct an associated Bogolyubov matrix to study the gravitational redshift of Gaussian states~\cite{2023JPhCS2531a2016A}. However, it was then shown that adding a single auxiliary mode \emph{cannot} yield a unitary matrix unless $\chi = 1$, see~\cite{Leber_2026}.
Furthermore, in the latter work it was also shown that \emph{if} a matrix of the form \eqref{eq:MMM block matrix} is unitary, then the added blocks must have at least the same dimensions as $\t u{_\chi}$. However, it was not proven that suitable choices of the sub-blocks in \cref{eq:MMM block matrix} actually lead to a unitary matrix. Moreover, due to the lack of explicit expressions, concrete calculations based on such matrices were effectively intractable. It is only because of the explicit dilation formula \eqref{redshift:matrix:dilation} proposed in this work that we were able to determine the Gaussian quantum channel \eqref{eq:Gaussian quantum channel result} analytically, and to provide an operational way of computing the relevant elements of the transformation matrices. In this sense, the fact that for all purposes it is not necessary to access any “perpendicular modes” makes our results readily applicable to any concrete scenario.

\section{Conclusions}
\label{s:conclusions}

We have analyzed the effects of the gravitational frequency shift of light within the framework of quantum field theory in stationary space-times. In contrast to the previously developed multi-mode mixer (\MMM) model, we demonstrated that the non-unitarity of the transformation matrices highlighted in previous work is due to the fact that the frequency-shift operator \eqref{eq:U operator definition} admits no finite-dimensional invariant subspaces (at least if one requires the field modes to have finite moments), regardless of the magnitude of the frequency shift.
To correct the issues of the \MMM\ model, we presented two alternatives:

\vspace{\baselineskip}

\noindent  \textit{Self-contained approach}: One option is to avoid using matrix representations in the first place.
    This is made possible by using an appropriate notation, which makes it explicit that ladder operators are in one-to-one correspondence with frequency spectra and that their commutation relations are determined by the Klein–Gordon products of these spectra. As a consequence, transition amplitudes for single-photon (and multi-photon) states can be computed as in \cref{eq:multi-mode transmission amplitudes}.
    This approach makes it evident that gravitational frequency shifts in stationary space-times preserve the photon number. Thus, the apparent non-unitarity found in the original version of the \MMM\ is merely a consequence of the projection of the frequency-shift operator onto a finite collection of modes.

\vspace{\baselineskip}

\noindent  \textit{Extension method}: An alternative approach that is conceptually closer to the original \MMM\ is to extend the non-unitary matrix $\t u{_i_j}(\chi)$ of transmission amplitudes to a unitary matrix on a larger space using \cref{redshift:matrix:dilation}. The key observation here is that the overall matrix \emph{cannot} be of the form $(\t f{_I}, \t U{_\chi} \t f{_J})$, where $\t f{_I}$ and $\t f{_J}$ range over field modes of interest and suitably chosen supplementary modes (no matter how they are chosen), for this would again lead to a non-unitary matrix. Instead, to achieve unitarity, the added entries must deviate from this form.
The unitary matrix \eqref{redshift:matrix:dilation} can then be used to apply standard quantum-optical and quantum-information-theoretical methods, as exemplified in \cref{s:Gaussian quantum channel}, where the gravitationally induced dynamics of a finite number of field modes is obtained by “tracing out” the supplementary degrees of freedom that were added in the dilation passing from $\t u{_\chi}$ to $\t\Upsilon{_\chi}$. 

The fact that the operator $\t u{_\chi}$ defined in \cref{eq:compression abstract} is a contraction implies that the effective Gaussian quantum channel derived in \cref{s:Gaussian quantum channel} is non-unitary. Experimentally, this could be verified by observing the fidelity of a Gaussian state compared to its frequency-shifted partner. Assuming the initial state to be pure, we demonstrated that reduction in the fidelity is observably only when the initial state is squeezed, which in turn can be seen as yet another confirmation that the observational signatures of gravitational frequency shifts depend on the quantum properties of the quantum state of light \cite{2021PhRvD.104h5015B,2023AnP...53500468B}.

We also introduced a measure of non-unitarity of the channel, termed the \emph{channel nonunitarity}, which can be computed using measurable quantities only. Contrary to the fidelity, this measure always detects deviations from unitarity irrespective of the initial state.

\vspace{\baselineskip}

Our construction of finite-dimensional unitary matrices associated to the gravitational frequency shift (of arbitrary magnitude) paves the way for further analyses of gravitational influences on light transmission in the language of quantum channels. This enables, for example, studies of the associated quantum channel capacity and its dependence on the gravitational as well as quantum features of the system, which is paramount to understanding relativistic effects in satellite-based quantum information protocols.

\vspace{-\baselineskip}
\section*{Acknowledgments}

We acknowledge Gerardo Adesso for useful suggestions.
TM was partly funded by the Deutsche Forschungsgemeinschaft (DFG, German Research Foundation) – project number \href{https://gepris.dfg.de/gepris/projekt/548528714}{548528714}.

\section*{Author contribution statement}

TM identified the source of non-unitarity in the \MMM\ model and developed the compression-approach for its correction, while the dilation-approach was suggested by DEB and AWS.
DEB proposed to study the effects using the quantum-channel formalism.
All authors contributed to writing the manuscript.

\bibliography{bibliography.bib}
\end{document}